\documentclass[nofootinbib,aps,11pt,preprintnumbers,unsortedaddress,prd]{revtex4-2}
\usepackage{amsmath,amssymb}
\usepackage{bbold}
\usepackage{yhmath}
\usepackage{subcaption}
\usepackage{graphicx}
\usepackage{hyperref}
\usepackage{mathtools}

\usepackage{tikz}
\usetikzlibrary{patterns}
\usetikzlibrary{shapes.geometric}
\usetikzlibrary{decorations.pathmorphing}
\usetikzlibrary{arrows.meta}
\tikzset{snake it/.style={decorate, decoration=snake}}
\usetikzlibrary{positioning}
\usetikzlibrary{arrows,shapes,positioning}
\usetikzlibrary{decorations.markings}
\tikzstyle arrowstyle=[scale=1]
\tikzstyle directed=[postaction={decorate,decoration={markings,mark=at position .65 with {\arrow[arrowstyle]{stealth}}}}]
\tikzstyle reverse directed=[postaction={decorate,decoration={markings,mark=at position .65 with {\arrowreversed[arrowstyle]{stealth};}}}]

\tikzset{->-/.style={decoration={
  markings,
  mark=at position #1 with {\arrow{>}}},postaction={decorate}}}

\tikzset{-<-/.style={decoration={
  markings,
  mark=at position #1 with {\arrow{<}}},postaction={decorate}}}

\newcommand{\mCH}{\ensuremath{\mathcal{CH}}}

\newcommand{\td}{\ensuremath{\,\text{d}}}
\DeclareMathOperator{\WF}{WF}
\DeclareMathOperator{\supp}{supp}
\newcommand{\bR}{\ensuremath{\mathbb{R}}}
\newcommand{\bC}{\ensuremath{\mathbb{C}}}
\newcommand{\bN}{\ensuremath{\mathbb{N}}}
\newcommand{\bS}{\ensuremath{\mathbb{S}}}

\newcommand{\norm}[1]{\ensuremath{\left\Vert #1 \right\Vert}}
\newcommand{\abs}[1]{\ensuremath{\left\vert #1 \right\vert}}

\usepackage{amsthm}
\newtheorem{thm}{Theorem}[section]
\newtheorem{prop}[thm]{Proposition}
\newtheorem{lem}[thm]{Lemma}
\newtheorem{cor}[thm]{Corollary}

\theoremstyle{definition}

\begin{document}
\title{Universality of the quantum energy flux at the inner horizon of asymptotically de Sitter black holes}

\author{Peter Hintz}
\email{peter.hintz@math.ethz.ch}
\affiliation{Department of Mathematics, ETH Zürich,\\ Rämistrasse 101, 8092 Zürich, Switzerland}
\author{Christiane K.M. Klein}
\email{christiane.klein@univ-grenoble-alpes.fr}
\affiliation{Institut für Theoretische Physik, Universität Leipzig,\\ Brüderstraße 16, 04103 Leipzig, Germany}
\affiliation{Univ. Grenoble Alpes, CNRS, IF, 38000 Grenoble, France}
\affiliation{AGM, CY Cergy Paris Université, 2 av. Adolphe Chauvin 95302 Cergy-Pontoise, France}

\begin{abstract}
    Recently, it was found that the energy flux of a free scalar quantum field on a Reissner--Nordstr\"om--de~Sitter spacetime has a quadratic divergence towards the inner horizon of the black hole. Moreover, the leading divergence was found to be state independent as long as the spectral gap of the wave equation on the spacetime is sufficiently large. In this work, we show that the latter result can be extended to all subextremal Reissner--Nordstr\"om--de~Sitter and subextremal Kerr--de~Sitter spacetimes with a positive spectral gap.
\end{abstract}
\maketitle

\section{Introduction}
The inner horizons of charged or rotating black holes pose an interesting problem with regard to determinism in general relativity. They are examples for the appearance of Cauchy horizons, or in other words smooth boundaries of the maximal Cauchy development of complete initial data. Since the boundary is smooth, one can in principle extend the spacetime beyond the boundary. However, such an extension is not unique, since it is not fixed by the initial data sufficient for the unique characterization of the spacetime up to that horizon. In this sense, determinism in general relativity is lost beyond the Cauchy horizon.

It has been argued by Penrose \cite{Penrose:1974} that this issue can be settled by noting that the Cauchy horizons of charged and rotating black holes are unstable. In other words, a generic perturbation of the black hole's initial data will make the spacetime inextendible across the Cauchy horizon. This has become known as the strong cosmic censorship conjecture (sCC). 

To make the conjecture more precise, one has to define the notions of inextendibility and of generic perturbations. To guarantee that any observer attempting to cross the horizon is inevitably destroyed by tidal deformations, the metric should be inextendible as a continuous function, as happens for example at the central singularity of a Schwarzschild black hole \cite{Sbierski:2018}. However, it has been shown in \cite{Dafermos:2017} that this version of sCC fails for Kerr black holes with non-zero angular momentum; this assumes the non-linear stability of the black hole exterior under small asymptotically flat perturbations, which is known for small angular momenta \cite{Klainerman:2023}.

Thus, one must content oneself with the slightly weaker version of the conjecture introduced by Christodoulou \cite{Christodoulou:2008}. In this version, inextendibility is understood in the sense of a weak solution to the Einstein equations. In other words, this version demands the inextendibility of the metric as a function in the Sobolev space $H^1_{\rm loc}$. While this will not cause the inevitable destruction of the careless observer, it still indicates a fundamental breakdown of the classical theory of general relativity. The (in-)validity of this version of the conjecture is still an open question, and a subject of active research.

A first step towards understanding the conjecture can be  made in two ways. The first one is to study it in a symmetric setting, e.g.\ by restricting to perturbations that maintain spherical symmetry; by Birkhoff's theorem, this necessitates coupling gravity to additional fields. This approach has been used for example in \cite{Dafermos:2003, Gajic:2017,Luk:2019a,Luk:2019b,Kehle:2023} on Reissner--Nordstr\"om spacetimes coupled to a scalar field. 

Alternatively, one can consider solutions to the (massless) scalar wave equation on a fixed background containing a Cauchy horizon as a first approximation of either linearized gravity or a simplified matter model coupled to gravity. It has been demonstrated that the linear perturbations indeed cease to be in $H^1_{\rm loc}$
on (parts of) the Cauchy horizon of Reissner--Nordstr\"om \cite{Luk:2015} and Kerr \cite{Dafermos:2015}.

Moreover, it has been shown \cite{Hintz:2015} that on Reissner--Nordstr\"om--de~Sitter (RNdS) and slowly rotating Kerr--de~Sitter (KdS) spacetimes solutions to the scalar wave equation are in $H^{1/2+\beta-0}$ near the Cauchy horizon. Here, $\beta=\alpha/\kappa_1$, where $\alpha$ is the spectral gap of quasi-normal modes, and $\kappa_1$ is the surface gravity of the Cauchy horizon. Hence, in this linear approximation, strong cosmic censorship is directly related to the spectral gap of the black hole's quasi-normal modes. Subsequent numerical studies of the quasi-normal modes \cite{Cardoso:2017} revealed that one can find $\beta>1/2$ for RNdS black holes of sufficiently large charge, and non-linear studies \cite{Hintz:2018, Hintz:2018b, Costa:2018, Luna:2019} (albeit the last two restrict to spherical symmetry) indicate that this remains valid, in other words $\beta$ remains the decisive quantity for determining the (in)stability of the Cauchy horizon when the non-linearities of the Einstein equations are taken into account. 

This raised the following question: since the scalar perturbations can be considered as part of a simple matter model, and since matter is, to our knowledge, most accurately described by quantum theory, can the inclusion of quantum effects potentially remedy the situation? To answer this question, the authors of \cite{Hollands:2019} study a free scalar quantum field on a fixed RNdS background spacetime, focusing on cases in which \cite{Cardoso:2017} had found $\beta$ to be larger than $1/2$. They find that the stress-energy tensor of the quantum scalar field in a generic Hadamard state can be split into a state-independent term and a state-dependent term. The state-dependent term is subsequently found to be in $L^{p}$, $p=(2-2\beta)^{-1}$, near the Cauchy horizon and therefore it does not diverge faster than the classical stress-energy tensor. The state-independent term is computed numerically in \cite{Hollands:2020} for a range of spacetime parameters. It has a generically non-vanishing quadratic divergence at the Cauchy horizon, in agreement with numerical results on Reissner--Nordstr\"om  and Kerr spacetimes \cite{Zilberman:2019, Zilberman:2022}. This quadratic divergence is stronger than the divergence of the classical stress-energy tensor, and therefore stronger than the state-dependent contribution, as long as the state is Hadamard and $\beta>1/2$. Consequently, the divergence of the quantum stress-energy tensor of the free scalar field at the Cauchy horizon of a RNdS spacetime satisfying $\beta>1/2$ is not only sufficiently strong to potentially remedy sCC, but also universal in the sense that it does not depend on the choice of state as long as the state is physically reasonable.

In this work we want to show that this universality holds more generally. We consider the stress-energy tensor of a free, real scalar quantum field near the Cauchy horizon of a RNdS or KdS spacetime merely satisfying $\beta>0$, i.e.\ having a positive spectral gap. We demonstrate that for a generic state $\omega$ which is Hadamard up to the Cauchy horizon, the expectation value of the stress-energy tensor can be split into a state-independent term  and a state-dependent term. 

The state-independent term can be computed numerically, and is expected to diverge quadratically, i.e.\ like  $(r-r_1)^{-2}$, at the Cauchy horizon, similar to the results obtained in  \cite{Hollands:2020,Zilberman:2019, Zilberman:2022}.
The state-dependent term is shown to diverge at most like $(r-r_1)^{-2+\beta}$ at the Cauchy horizon, where $r$ is the usual radial coordinate of these spacetimes and $r_1$ the radius of the Cauchy horizon. (Strictly speaking, we show an upper bound of $(r-r_1)^{-2+\beta'}$ for some $0<\beta'<\beta$, though conjecturally one can take $\beta'$ arbitrarily close to $\beta$.)

In other words, we show that under two conditions---first, $\beta>0$, i.e.\ there is a positive spectral gap and mode stability holds, and second, the state-independent term is generically non-vanishing---the state-independent term constitutes the leading divergence of the quantum stress-energy tensor at the Cauchy horizon, extending the universality result to all $\beta>0$. Moreover, under the same conditions, the universal leading divergence of the quantum stress-energy tensor at the Cauchy horizon is stronger than the divergence of the classical stress-energy tensor. The condition $\beta>0$ is known to hold for scalar fields on all subextremal RNdS spacetimes (for zero and non-zero scalar field masses) as well as on KdS spacetimes which are either slowly rotating \cite{Dyatlov:2011} or have small mass \cite{Hintz:2021} (for zero and also for small non-zero scalar field masses by \cite[Lemma~3.5]{HintzVasy:2015}), and it is conjectured to hold in the full subextremal KdS range \cite{Yoshida:2010,Hatsuda:2020}.

This leads to the conclusion that perturbations caused by quantum effects should not be neglected in considerations of Cauchy horizon stability, since they will eventually become comparable in size to, and ultimately dominate classical perturbations as the Cauchy horizon is approached. To estimate how close to the Cauchy horizon this happens, one can, for example, compare the $rr$-components of the stress-energy tensor of the classical and quantum scalar field near the Cauchy horizon under the assumption that the leading contribution to the quantum result is non-vanishing. Taking into account the  pointwise bounds for the classical stress-energy tensor, and the relative size of quantum and classical results away from the horizon, which should typically be of order $r_P/L$, one finds that the $rr$-component of the quantum stress-energy tensor will become dominant when $r-r_1\ll r_P(r_P/L)^{1/\beta-1}$, where $L$ is a typical length scale of the spacetime, for example $M$.

We will begin our discussion with an introduction of the geometric and field-theoretic setup in section \ref{sec:setup}. We will also use this section to recall some results on Hadamard states. Section \ref{sec:reg from decay}  will demonstrate how the decay of the solutions of the wave equations towards $i^+$ can be translated into an estimate on the behavior at the Cauchy horizon. This estimate will be employed in section \ref{sec:bound} to bound the divergence of the state-dependent part of the expectation value of the stress-energy tensor in a generic Hadamard state $\omega$. Concluding remarks will be given in section \ref{sec: Conclusion}.

\section{Setup}
\label{sec:setup}
\subsection{The RNdS and KdS spacetimes}
In this work, we will consider subextremal RNdS and KdS spacetimes. They can be described by a metric of the form 
\begin{align}
\label{eq:RNdS_g_BL}
   g^{\text{RNdS}}_{\lambda, Q}=-\frac{\Delta_r^{\text{RNdS}}}{r^2}\td t^2+\frac{r^2}{\Delta_r^{\text{RNdS}}}\td r^2+r^2\td\Omega^2\, ,
\end{align}
with 
\begin{align}
   \Delta_r^{\text{RNdS}}=-\lambda r^4+r^2-2Mr+Q^2 
\end{align}
for RNdS and
\begin{align}
\label{eq:KdS_g_BL}
    g^{\text{KdS}}_{\lambda, a}&=\frac{\Delta_\theta a^2\sin^2\theta-\Delta_r^{\text{KdS}}}{\rho^2\chi^2}\td t^2+\left[\Delta_\theta(r^2+a^2)^2-\Delta_r^{\text{KdS}}a^2\sin^2\theta\right]\frac{\sin^2\theta}{\rho^2\chi^2}\td\varphi^2\\
\nonumber &+\frac{\rho^2}{\Delta_r^{\text{KdS}}}\td r^2+\frac{\rho^2}{\Delta_\theta}\td \theta^2+2\frac{a\sin^2\theta}{\rho^2\chi^2}[\Delta_r^{\text{KdS}}-\Delta_\theta(r^2+a^2)]\td t\td\varphi\, ,
\end{align}
with
\begin{align}	
\label{eq:D_t,rho,chi}
    \Delta_r^{\text{KdS}}&=(1-\lambda r^2)(r^2+a^2)-2Mr\, , &  \Delta_\theta &=  1 + a^2 \lambda \cos^2\theta\, , \\ \rho^2&=r^2+a^2\cos^2\theta\, , &  \chi&=1+a^2\lambda\, ,
\end{align}
for KdS. Throughout this discussion, we will choose the scale for the coordinates $r$ and $t$ such that the black hole mass $M$ is set to one. The black hole's charge $Q$ or angular momentum parameter $a$ as well as the cosmological constant $\Lambda=3 \lambda$ are chosen such that the functions $\Delta_r^{\#}$, with $\#$ replacing either "RNdS" or "KdS", have three real distinct positive roots $r_1<r_2<r_3$ indicating the locations of the cosmological horizon $(r_3)$, the outer $(r_2)$ and the inner horizon $(r_1)$ of the black hole. 

The coordinate singularities at the horizons can be eliminated by introducing advanced or retarded time coordinates defined by
\begin{align}
\label{eq:tpm}
    \td t_\pm =\td t\pm \frac{\chi (r^2+a^2)}{\Delta_r^{\#}}\td r\, ,
\end{align}
where we set $\chi=1$ and $a=0$ for RNdS. For Kerr--de~Sitter, one has to introduce in addition the azimuthal coordinates
\begin{align}
\label{eq:phipm}
    \td \varphi_\pm =\td \varphi \pm \frac{\chi a }{\Delta_r^\#}\td r\, .
\end{align}
The coordinates $(t_\pm, r, \theta,\varphi_\pm)$ then allow an extension of the metric through the outgoing/ingoing pieces of the horizons. It takes the form
\begin{align}
    g^\#=g^\#_{tt} \td t_\pm^2\pm 2\frac{1}{\chi}\td t_\pm\td r+g^\#_{\theta\theta}\td \theta^2+g^\#_{\varphi\varphi}\td\varphi_\pm^2\mp 2\frac{a\sin^2\theta}{\chi}\td\varphi_\pm\td r\, ,
\end{align}
where again, we set $\chi=1$ and $a=0$ for the RNdS case, and $g^\#_{\mu\nu}$ are the components of the corresponding metric in the Boyer-Lindquist coordinates as given in \eqref{eq:RNdS_g_BL} and \eqref{eq:KdS_g_BL}.

It should be mentioned that none of these coordinate systems cover the axis where $\sin\theta=0$. However, it has been shown that the metric can be analytically extended to the axis as well using a suitable coordinate transformation \cite{Hintz:2015, Borthwick:2018}.

The physical RNdS and KdS spacetimes respectively will be the manifolds $\bR_{t_+}\times(r_1,r_3)_r\times\bS^2_{(\theta,\varphi_+)}$ glued to $\bR_{t_-}\times(r_2,\infty)_r\times \bS^2_{(\theta,\varphi_-)}$ on $\{r_2<r<r_3\}$, equipped with the metric $g^{\text{RNdS}}_{\lambda,Q}$ or $g^{\text{KdS}}_{\lambda,a}$. More details on the advanced/retarded time coordinates, the extended spacetime and the gluing can be found in \cite{Mokdad:2017, Borthwick:2018}. We will refer to the physical spacetime as $\mathcal{M}$; see also Figure \ref{fig: Penrose diagram full}. 

The physical RNdS  and KdS spacetimes are globally hyperbolic. For RNdS, this follows from the analysis in \cite{Mokdad:2017}, combined with the fact that the physical spacetime $\mathcal{M}$ is a causally convex subset of the union of Kruskal domains around $r=r_2$ and $r=r_3$ discussed in \cite{Mokdad:2017}. For KdS, this was shown by an explicit construction in \cite{Klein:2022}.

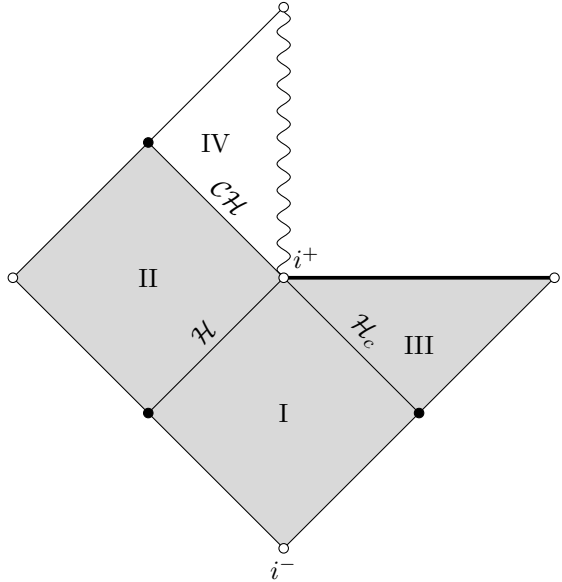
\begin{figure}
    \centering
\begin{tikzpicture}[scale=0.9]
\path[fill=gray!30](0,0)--(-4,4)--(-2,6)--(0,4)--(4,4)--(0,0);
 \draw (0,0) -- (-2,2) ;
 \draw (-2,2) -- (-4,4) ;
 \draw (0,0) -- (2,2) ;
 \draw (2,2) -- (0,4) node[midway,above,sloped]{$\mathcal{H}_c$};
 \draw (-2,2) -- (0,4) node[midway, above, sloped]{$\mathcal{H}$};
 \draw (0,4) -- (-2,6) node[midway, above, sloped]{$\mathcal{C}\mathcal{H}$};
 \draw (-4,4) -- (-2,6);

 \draw[snake it] (0,4) -- (0,8);
 \draw (-2,6) -- (0,8);
 \draw[double=black] (0,4) -- (4,4); 
 \draw (2,2) -- (4,4);
 \draw (0,2) node{${\rm I}$};
 \draw (2,3) node{${\rm III}$};
 \draw (-2,4) node{${\rm II}$};
 \draw (-1,6) node{${\rm IV}$};

 \draw[fill=white] (-4,4) circle (2pt);
 \draw[fill=white] (4,4) circle (2pt);
 \draw[fill=white] (0,8) circle (2pt);
 \draw[fill=black] (-2,6) circle (2pt);
 \draw[fill=white] (0,0) circle (2pt) node[below]{$i^-$};
 \draw[fill=black] (-2,2) circle (2pt);
 \draw[fill=white] (0,4) circle (2pt) node[above right]{$i^+$};
 \draw[fill=black] (2,2) circle (2pt);
  \end{tikzpicture}
  \caption{The Penrose diagram of the subextremal RNdS spacetime, or the Carter--Penrose diagram for the subextremal KdS spacetime. The gray region indicates our physical spacetime $\mathcal{M}$, while the diagram shows also the analytic extension across $\mCH$.}
  \label{fig: Penrose diagram full}
\end{figure}  

In this work, we focus mostly on the Cauchy horizon, which is the future boundary of $\mathcal{M}$ considered as a submanifold of its analytic extension. More specifically, we focus on the ingoing part of the Cauchy horizon, cf.\ Figure \ref{fig: Penrose diagram full}. The metric can be analytically extended through the ingoing piece of the  Cauchy horizon in the retarded coordinates $(t_-,r,\theta,\varphi_-)$.

In Section \ref{sec:reg from decay} we will analyse the classical wave equation on a domain $\Omega\subset \mathcal{M}$ which encompasses the relevant part of the Cauchy horizon. More specifically, the domain $\Omega$ is bounded in the past by a hypersurface of constant $r$, and in the future by a piece of the Cauchy horizon as well as a spacelike hypersurface transversal to the Cauchy horizon, as indicated in Figure \ref{fig: Omega}.

\begin{figure}
    \centering
    \includegraphics{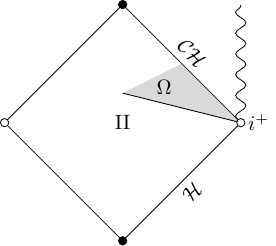}
    \caption{Illustration of the domain $\Omega$, in which the results of decay towards $i^+$ are propagated and converted to regularity results at $\mCH$.}
    \label{fig: Omega}
\end{figure}

\subsection{The free scalar field theory}
\label{subsec: scalar field}
We consider a free scalar field on the physical RNdS or KdS spacetime $\mathcal{M}$ satisfying the Klein--Gordon equation
\begin{align}
\label{eq:KGE}
    \mathcal{P}\phi=0\, ,\quad \mathcal{P}=\Box_g-\left(m^2+\xi R\right)\, ,
\end{align}
where $\Box_g$ is the d'Alembert operator, $R=4\Lambda>0$ is the Ricci scalar of RNdS or KdS respectively, and $m\geq 0$, $\xi\geq 0$ are constants. The quantum theory for this field can be described by the CCR-algebra $\mathcal{A}(\mathcal{M})$ of quasi-local observables, which can be defined as the free *-algebra generated by the identity $\mathbf{1}$ and the smeared field operators $\Phi(f)$, $f\in \mathcal{C}_0^\infty$, subject to the relations
\begin{itemize}
\item  $\Phi(\alpha f+ g)=\alpha \Phi(f)+ \Phi(g)\quad \forall f,g\in \mathcal{C}_0^\infty(\mathcal{M})$, $\alpha \in \bC$
\item  $\Phi(\mathcal{P}f)=0 \quad\forall f\in \mathcal{C}_0^\infty(\mathcal{M})$
\item  $(\Phi(f))^*=\Phi(\bar f)\quad \forall f\in \mathcal{C}_0^\infty(\mathcal{M})$
\item  $\left[\Phi(f),\Phi(g)\right]=iE(f,g)\mathbf{1} \quad\forall f,g\in \mathcal{C}_0^\infty(\mathcal{M})$ .
\end{itemize}
Here, $E$ is the commutator function, or Pauli--Jordan propagator of $\mathcal{P}$. It is constructed as the difference between the unique retarded and advanced Green's operators $E^\pm$ for the Klein--Gordon operator $\mathcal{P}$. 

A state in this framework is a linear map $\omega:\mathcal{A}(\mathcal{M})\to \bC$ satisfying $\omega(\mathbf{1})=1$ and $\omega(A^*A)\geq 0$ for all $A\in \mathcal{A}(\mathcal{M})$. 
It is called a quasi-free Hadamard state if it is entirely determined by its two-point function
\begin{align*}
    w(f,h)=\omega(\Phi(f)\Phi(h))\, ,
\end{align*}
and if the wavefront set of the two-point function, considered as a distribution on $\mathcal{M}\times\mathcal{M}$, satisfies the microlocal spectrum condition \cite{Radzikowski:1996}
\begin{align}
\label{eq:muloc spec cond}
\WF^\prime(w)&=\mathcal{C}^+\, ,\\
\mathcal{C}^\pm&=\{(x,k;y,l)\in T^*(\mathcal{M}\times \mathcal{M}):(x,k)\sim (y,l)\text{ and } \pm k \rhd 0\}\, .
\label{eq:Cpm}
\end{align}
Here, $(x,k;y,l)\in \WF^\prime(w)$ if and only if $(x,k;y,-l)\in \WF(w)$. A point $(x,k)\in T^*\mathcal{M}$ is related to $(y,l)\in T^*\mathcal{M}$ by $\sim$, $(x,k)\sim (y,l)$, if $x$ and $y$ are connected by a null geodesic to which $k$ is the cotangent vector at $x$, and $l$ agrees with $k$ coparallel transported along the geodesic to $y$. In other words, $(x,k)\sim(y,l)$ if the two points lie in the same bicharacteristic strip of $\mathcal{P}$. The notation $k\rhd 0 $ means that $k(v)>0$ for all time-like future-pointing vectors in $T_x\mathcal{M}$, i.e.\ $k$ is a future-pointing covector. Physically reasonable states are usually required to satisfy the microlocal spectrum condition, since it allows one to extend the algebra to include important observables such as the (smeared) stress-energy tensor \cite{Hollands:2001, Hollands:2001b}, and it results in finite expectation values with finite variance for these observables \cite{Brunetti:1999, Fewster:2013}. The condition that the state is quasi-free is made here for simplification of notation in the microlocal spectrum condition, and should not influence the following arguments, since we will be mainly interested in the two-point function.

The most relevant property of Hadamard states for our purpose is the following: consider any two Hadamard states $\omega$, $\omega^\prime$ on the CCR-algebra $\mathcal{A}(M)$ of any globally hyperbolic spacetime $(M,g)$. Then the difference of two-point functions considered as a distribution on $M\times M$,
\begin{align}
    W[\omega, \omega^\prime]\in \mathcal{D}^\prime(M\times M)\, ,\quad W[\omega, \omega^\prime](f,h)=\omega(\Phi(f)\Phi(h))-\omega^\prime(\Phi(f)\Phi(h)),\quad f,h\in \mathcal{C}_0^\infty(M),
\end{align}
is of the form
\begin{align}
    W[\omega,\omega^\prime](f,h)=\int\limits_{M\times M}W[\omega,\omega^\prime](x,y)f(x)h(y)\td \mu(x)\td \mu(y)\, ,
\end{align}
with $\td\mu$ the volume form induced on $M$ by the metric $g$ and $W[\omega,\omega^\prime](x,y)\in \mathcal{C}^\infty(M\times M)$. Moreover, this function is a real, symmetric bi-solution to the Klein--Gordon equation,
\begin{align*}
  \mathcal{P}(x)W[\omega,\omega^\prime](x,x^\prime)=\mathcal{P}(x^\prime)W[\omega,\omega^\prime](x,x^\prime)=0\, .  
\end{align*}
An important consequence of the smoothness becomes apparent when one considers the (smeared) local, non-linear observables of the theory. These observables, which lie in the extension of the CCR-algebra $\mathcal{A}(M)$, can be written as linear combinations of locally and covariantly renormalized Wick powers of differentiated fields ${:}\left(\prod_{i=1}^n\mathcal{D}_i\Phi\right)(f){:}$ \cite{Hollands:2001}. Here, $\mathcal{D}_i$ are (not necessarily scalar) differential operators, $f\in \mathcal{C}_0^\infty(M)$ is a smearing function, and the double dots indicate that this quantity has been renormalized, with all renormalization ambiguities fixed in some way. Focusing on the case $n=2$, which encompasses relevant observables such as the stress-energy tensor, the smoothness of $W[\omega,\omega^\prime](x,x^\prime)$ allows us to write
\begin{align}
\label{eq:point-split}
    \omega({:}\left(\mathcal{D}_1\Phi\mathcal{D}_2\Phi\right)(x){:})-\omega^\prime({:}\left(\mathcal{D}_1\Phi\mathcal{D}_2\Phi\right)(x){:})=\lim\limits_{x^\prime\to x}\left(g(x,x^\prime) \mathcal{D}_1(x) \mathcal{D}_2(x^\prime) W[\omega,\omega^\prime](x,x^\prime)\right)\, ,
\end{align}
where ${:}\left(\mathcal{D}_1\Phi\mathcal{D}_2\Phi\right)(x){:}$ should be understood as operator-valued distributions, and $g(x,x^\prime)$ is the proper power of the parallel transport bi-tensor $g_\mu^{\nu^\prime}(x,x^\prime)$ mapping $T_{x^\prime}M$ to $T_x M$, so that the result is a tensor at $x$ in the case when the derivative operator $\mathcal{D}_2$ is not scalar. The right hand side is the coinciding point limit of a $\mathcal{C}^\infty$-function on $M\times M$, see \cite[H1)-H3)]{Hollands:2019}, and hence a smooth function on $M$. Therefore, the expression on the left hand side, which should a priori be understood in a distributional sense, is a smooth function as well. Indeed, it follows from the conditions in the local and covariant renormalization scheme that also $\omega({:}\Phi^k(x){:})$ is a smooth function on $M$ as long as $\omega$ is a Hadamard state \cite{Hollands:2001}, implying that one can discuss the expectation values of Wick powers without smearing. These results will become crucial in Section~\ref{sec:bound}.

\section{Expansion of classical solutions}
\label{sec:reg from decay}

Our aim in this section is to prove sharp pointwise bounds for scalar fields near the Cauchy horizon of subextremal RNdS and KdS spacetimes (see Corollary~\ref{cor:wave-CH}). We begin by discussing the massless scalar wave equation on a fixed subextremal KdS background $g=g_{\lambda,a}^{\rm KdS}$; see \eqref{eq:KdS_g_BL}. We write $\Delta_r=\Delta_r^{\rm KdS}$. Passing to the coordinates $t_\pm$ and $\varphi_\pm$ from~\eqref{eq:tpm}--\eqref{eq:phipm} amounts to replacing $\partial_t,\partial_\varphi$, and $\partial_r$ by $\partial_{t_\pm}$, $\partial_{\varphi_\pm}$, and $\partial_r\pm\frac{\chi}{\Delta_r}((r^2+a^2)\partial_{t_\pm}+a\partial_{\varphi_\pm})$, respectively, so the wave operator
\[
  \rho^2\Box_g = -\frac{\chi^2}{\Delta_r}\bigl((r^2+a^2)\partial_t+a\partial_\varphi)^2 + \frac{\chi^2}{\Delta_\theta\sin^2\theta}(a\sin^2\theta\,\partial_t+\partial_\varphi)^2 + \partial_r \Delta_r \partial_r + \frac{1}{\sin\theta}\partial_\theta\Delta_\theta\sin\theta\,\partial_\theta
\]
becomes
\begin{equation}
\label{eq:Boxpm}
\begin{split}
  \rho^2\Box_g &= \partial_r\Delta_r\partial_r \pm \chi\bigl((r^2+a^2)\partial_{t_\pm}+a\partial_{\varphi_\pm}\bigr)\partial_r \pm \partial_r \chi\bigl((r^2+a^2)\partial_{t_\pm}+a\partial_{\varphi_\pm}\bigr) \\
    &\qquad + \frac{\chi^2}{\Delta_\theta\sin^2\theta}(a\sin^2\theta\,\partial_{t_\pm}+\partial_{\varphi_\pm})^2 + \frac{1}{\sin\theta}\partial_\theta\Delta_\theta\sin\theta\,\partial_\theta.
\end{split}
\end{equation}
Recall that the coordinates $t_+,\varphi_+$ are valid in the union of the regions I and II as well as the future event horizon $\mathcal H$ in Figure~\ref{fig: Penrose diagram full}, while $t_-,\varphi_-$ are valid in the union of the regions I and III as well as the cosmological horizon $\mathcal{H}_c$, and also in the union of the regions II and IV as well as the Cauchy horizon $\mathcal{CH}$; the level sets of $t_+$ are transversal to the future event horizon, and the level sets of $t_-$ are transversal to the cosmological horizon (in regions I and III) and to the Cauchy horizon (in regions II and IV).

Let $t_*$ denote a time function in the union of regions I, II, and III which in regions I and II differs from $t_+$ by a smooth function of $r\in(r_1,r_3)$, and which in regions I and III differs from $t_-$ by a smooth function of $r\in(r_2,\infty)$. We may choose such a function $t_*$ to have spacelike level sets (such as $\Sigma_-$ and $\Sigma_+$ in Figure~\ref{fig: illustration}). We now recall:

\begin{thm}[\cite{Petersen:2021}]
\label{thm:PetersenVasy}
  Let $r_-\in(r_1,r_2)$ and $r_+>r_3$. Write $e^{-\alpha t_*}H^s$ for the space of functions $\psi=\psi(t_*,x)$ (where $x\in\mathbb{R}^3$ denotes Cartesian coordinates on $(r_-,r_+)\times\mathbb{S}^2\subset\mathbb{R}^3$) with support in $t_*\geq 0$ so that
  \[
    \|\psi\|_{e^{-\alpha t_*}H^s}^2 := \sum_{j+|\beta|\leq s} \|e^{\alpha t_*}\partial_{t_*}^j\partial_x^\beta \psi\|_{L^2}^2 < \infty,
  \]
  where $\|\cdot\|_{L^2}$ is the spacetime $L^2$-norm. Then there exists $\alpha_1>0$ so that the following holds. Let $s>\frac12+\alpha_1\max(\frac{1}{\kappa_2},\frac{1}{\kappa_3})$, where $\kappa_j$ denotes the surface gravity of the horizon $r=r_j$. Then for $\alpha'<\alpha_1$ and for all $b\in e^{-\alpha' t_*}H^s$, the unique retarded solution of $\Box_g\psi=b$ has an asymptotic expansion
  \begin{equation}
  \label{eq:expansion}
      \psi - \sum_{j=1}^N\sum_{k=0}^{k_j-1} t_*^k e^{-i\sigma_j t_*}v_{j k} = \tilde\psi \in e^{-\alpha' t_*}H^s,
  \end{equation}
  where the $\sigma_1,\ldots,\sigma_N$ are the finitely many quasinormal modes with $\Im\sigma_j>-\alpha_1$, $k_j$ is the multiplicity of $\sigma_j$, and $e^{-i\sigma_j t_*}\sum_{k=0}^{k_j} t_*^k v_{j k}$ is a corresponding (smooth) mode solution. Furthermore, $\|\tilde\psi\|_{e^{-\alpha' t_*}H^s}\leq C\|b\|_{e^{-\alpha' t_*}H^s}$ for some constant $C$ depending only on $r_-,r_+,\alpha',s$, and the KdS black hole parameters.
\end{thm}

\begin{thm}
\label{thm:mode-stability}
  In the notation of Theorem~\ref{thm:PetersenVasy}, mode stability holds---that is, the quasinormal mode $\sigma_1=0$ has $k_1=1$ and corresponding mode solution equal to a constant, and all other $\sigma_j$ have $\Im\sigma_j<0$---under either one of the following two conditions.
  \begin{enumerate}
  \item The KdS black hole is slowly rotating, i.e.\ $0<9\Lambda M^2<1$ and $|a/M|\leq C(\Lambda M^2)$ where $C\colon(0,1)\to(0,\infty)$ is a positive continuous function \cite{Dyatlov:2011,Dyatlov:2012}.
  \item The KdS black hole has a small mass, i.e.\ $|a/M|\in[0,1)$ and $0<\Lambda M^2\leq C(|a/M|)$ where $C\colon[0,1)\to(0,\infty)$ is a positive continuous function \cite{Hintz:2021}.
  \end{enumerate}
\end{thm}

The combination of the two results implies that for source terms $b$ which vanish for large $t_*$, the exponential decay rate of $\psi$ (in an $L^2$-sense, but via Sobolev embedding for $s>2+k$ also in the $\mathcal{C}^k$ sense, i.e.\ in a pointwise sense with up to $k$ derivatives) towards a constant is at least $\alpha-\epsilon$ for all $\epsilon>0$ where $\alpha$ is the spectral gap, i.e.\ the infimum of $-\Im\sigma$ over all non-zero quasinormal modes $\sigma\in\mathbb{C}$.\footnote{There is a technical subtlety here: the constant $\alpha_1$ for which Theorem~\ref{thm:PetersenVasy} is proved is given in terms of dynamical quantities associated with the trapped set. We are thus implicitly assuming that $\alpha_1\geq\alpha$. This is valid in the second setting described by Theorem~\ref{thm:mode-stability} and proved in the given reference. It is also valid in the Schwarzschild--de~Sitter case $a=0$. In the slowly rotating Kerr--de Sitter case, it is true as well, and follows from the validity of a \emph{full} resonance expansion with error terms having any desired amount of exponential decay, as demonstrated in \cite{Dyatlov:2012}, except in this case the remainder $\tilde\psi$ lies in $e^{-\alpha' t_*}H^{s-d}$ where $d$ depends on $\alpha'$ (in \cite{Dyatlov:2012} an estimate is stated only for $s-d=1$, but the arguments given there give the claimed stronger statement).}

If one wishes to consider the full subextremal range of KdS parameters, one needs to \emph{assume} the validity of mode stability; while this is not known rigorously, there is strong numerical support \cite{Yoshida:2010,Hatsuda:2020}.\footnote{In the present work, we only use exponential decay to constants; it is not necessary to know that the exponential decay rate is exactly given by the spectral gap. This is important since Theorem~\ref{thm:PetersenVasy} does not give this more precise (and conjecturally true, but as of yet unproven) information.}

We thus proceed under the assumption---which as mentioned above is satisfied in the settings of Theorem~\ref{thm:mode-stability}, and conjecturally in the full subextremal KdS range---that
\begin{subequations}
\begin{equation}
\label{eq:wave-decay}
  \Box_g\psi = b \in e^{-\alpha' t_*}H^{s+d} \implies \psi = c + \tilde\psi,\quad c\in\mathbb{C},\ \tilde\psi\in e^{-\alpha' t_*}H^s,
\end{equation}
for some $\alpha'>0$ (which must satisfy $\alpha'\leq\alpha$), for all sufficiently large $s$, and some fixed $d\geq 0$, on the spacetime region where $(r_1,r_2)\ni r_-<r<r_+\in(r_3,\infty)$; and
\begin{equation}
\label{eq:wave-decay-est}
  |c| + \|\tilde\psi\|_{e^{-\alpha' t_*}H^s} \leq C\|b\|_{e^{-\alpha' t_*}H^{s+d}}.
\end{equation}
\end{subequations}

For the purposes of the present paper, we only need to consider source terms $b$ whose support is a compact subset of $\mathbb{R}_{t_*}\times(r_1,r_3)\times\mathbb{S}^2$ (see Figure~\ref{fig: illustration}); after a constant shift of $t_*$, we shall thus only consider $b$ which vanish for $t_*\geq 1$. Turning attention to the black hole interior, we record that the solution $\psi$ in~\eqref{eq:wave-decay} thus satisfies $\Box_g\psi=0$ in $t_*\geq 1$ and $\psi-c\in e^{-\alpha' t_*}H^s([1,\infty)\times(r_-,r_\sharp)\times\mathbb{S}^2)$ where $r_1<r_-<r_\sharp<r_2$. The following is the main technical result of this section.

\begin{prop}
\label{prop:wave-CH}
  Let $r_1<r_\flat<r_\sharp<r_2$. Fix a smooth function $U\colon[r_1,r_2)\to\mathbb{R}$ so that the level sets of $u:=t_-+U(r)$ are spacelike. Let $u_\sharp$ denote the $u$-coordinate of the point $t_*=1,r=r_\sharp$, and define the domain $\Omega=(u_\sharp,\infty)_u\times(r_1,r_\sharp)_r\times\mathbb{S}^2$ inside the KdS spacetime. (See Figure~\ref{fig: Omega}.) Let $u_\flat>u_\sharp$. Let $\alpha'>0$, and suppose $\psi$ is a solution of $\Box_g\psi=0$ on $\Omega$ which is of the form
  \begin{equation}
  \label{eq:psi-structure}
    \psi=c+\tilde\psi,\qquad c\in\mathbb{C},\quad \tilde\psi|_{\Omega'}\in e^{-\alpha' u}H^s(\Omega')
  \end{equation}
  where $\Omega'=\Omega\cap\{r_\flat<r<r_\sharp\}$ and $s>\frac52+\frac{\alpha'}{\kappa_1}+m$, $m\in\mathbb{N}_0$. Then there exists a function
  \[
    \psi_0=\psi_0(u,\omega)\in \mathcal{C}^m((u_\sharp,u_\flat)\times\mathbb{S}^2)
  \]
  so that for all $j,k\in\mathbb{N}_0$ and $\gamma\in\mathbb{N}_0^2$ with $j+k+|\gamma|\leq m$ we have the pointwise bound\footnote{If $\frac{\alpha'}{\kappa_1}=\ell+\delta$ with $\ell\in\mathbb{N}$ and $\delta\in(0,1]$, the 0-th order Taylor expansion of $\psi(r,u,\omega)=\psi_0(u,\omega)+\ldots$ at $r=r_1$ can be improved to an $\ell$-th order expansion with a remainder term whose $j$-th $r$-derivative is of size $\mathcal{O}((r-r_1)^{\frac{\alpha'}{\kappa_1}-j})$. Since in our application it only matters that $\alpha'>0$, we content ourselves with the stated version. An analogous comment applies to Corollary~\ref{cor:wave-CH} below when $\beta>1$.}
  \begin{equation}
  \label{eq:wave-CH-est}
    |\partial_r^j\partial_u^k\partial_\omega^\gamma\bigl(\psi(r,u,\omega) - \psi_0(u,\omega)\bigr)| \leq C_{j k\gamma}(r-r_1)^{\min(\frac{\alpha'}{\kappa_1},1)-j}\|\tilde\psi|_{\Omega'}\|_{e^{-\alpha' u}H^s(\Omega')}.
  \end{equation}
\end{prop}

We remark that Proposition~\ref{prop:wave-CH} is different from (and less delicate than) scattering theory from the event horizon (rather than from a hypersurface in the black hole interior) to the Cauchy horizon in the black hole interior, as studied for example in \cite{Kehle:2019}. Proposition~\ref{prop:wave-CH} has the following immediate consequence, which will be used as a black box in Section~\ref{sec:bound}:

\begin{cor}[Pointwise bounds near the Cauchy horizon of KdS]
\label{cor:wave-CH}
  Set $\beta'=\frac{\alpha'}{\kappa_1}$ (which satisfies $0<\beta'\leq\beta=\frac{\alpha}{\kappa_1}$). Fix $T_0<T_1$, $r_\sharp\in(r_1,r_2)$, $r_+\in(r_3,\infty)$, and $m\in\mathbb{N}_0$. Fix further $u_\sharp<u_\flat$ in the notation of Proposition~\ref{prop:wave-CH}. Then there exist $m'\in\mathbb{N}$ and a constant $C$ so that for all $b\in \mathcal{C}^{m'}(\mathbb{R}_{t_*}\times(r_1,\infty)\times\mathbb{S}^2)$ with support in $\{T_0\leq t_*\leq T_1,\ r_\sharp\leq r\leq r_+\}$, the retarded solution of $\Box_g\psi=b$ satisfies
  \begin{align}
    |\partial_u^k\partial_\omega^\gamma\psi(r,u,\omega)| &\leq C_{k\gamma}\|b\|_{\mathcal{C}^{m'}}, \\
    |\partial_r^j\partial_u^k\partial_\omega^\gamma\psi(r,u,\omega)| &\leq C_{j k\gamma \epsilon}(r-r_1)^{\min(\beta',1)-\epsilon-j}\|b\|_{\mathcal{C}^{m'}}
  \end{align}
  in the region $\{u_\sharp<u<u_\flat,\ r_1<r<r_\sharp\}$ for all $\epsilon>0$ and for all $j,k,\gamma$ with $j+k+|\gamma|\leq m$. 
\end{cor}

\begin{proof}[Proof of Proposition~\ref{prop:wave-CH}]
  By subtracting from $\psi$ the constant $c$ (which solves the wave equation), we may assume that $\psi=\tilde\psi$. We sketch two different proofs of the estimate~\eqref{eq:wave-CH-est}; the first one closely follows arguments from \cite{Hintz:2015}, whereas the second one is more direct and can in principle be extended to produce more precise asymptotic expansions at $r=r_1$, given asymptotic expansions (resonance expansions, or even expansions into powers of $t_*$ as in Price's law on Kerr) of $\psi$ in $r_\flat\leq r\leq r_\sharp$.
  
  \textbf{First proof.} We can adapt the methods of \cite{Hintz:2015}; in fact, the following arguments ultimately give a simpler proof of the main results of \cite{Hintz:2015}, in that the structure of spacetime in the region $r\geq r_2$, or indeed $r\geq r_\sharp\in(r_1,r_2)$, plays no role, \emph{given the a priori assumption}~\eqref{eq:psi-structure} on the structure of the solution of the wave equation.\footnote{This a priori assumption in turn follows from analysis in a neighborhood of $\{r_2\leq r\leq r_3\}$, cf.\ Theorems~\ref{thm:PetersenVasy} and \ref{thm:mode-stability} above, which can be done completely independently of the analysis in the black hole interior.} To wit, we work in the region $r\leq r_\sharp$ and as in the reference consider the wave equation on an artificial extension of the KdS spacetime to $r<r_1$ which features yet another artificial horizon at some value $r=r_0<r_1$; and we place a time-translation-invariant complex absorbing operator  $\mathcal{Q}$ in the region $r<r_1$. In the notation of \cite{Hintz:2015}, we thus work in
  \[
    \Omega_{\rm ext} := [0,\infty)_{t^*}\times[r_0-2\delta,r_\sharp]\times\mathbb{S}^2
  \]
  (where we write $t^*$ for the function denoted $t_*$ in~\cite{Hintz:2015} to distinguish it from the time function $t_*$); and for $r_1\leq r<r_2$, the difference $t^*-t_-$ is a smooth function of $r$, and thus so is $t^*-u$.
  
  Let $\zeta=\zeta(r)$ denote a smooth function which equals $1$ near $(-\infty,r_\flat]$ and $0$ near $[r_\sharp,\infty)$, then $\zeta\psi$ satisfies the equation
  \[
    \Box_g(\zeta\psi)=b':=[\Box_g,\zeta]\psi\in e^{-\alpha' t_*}H^{s-1}(\Omega_{\rm ext})
  \]
  in $r>r_1$, with $b$ vanishing outside $\{r_\flat\leq r\leq r_\sharp\}$. Following the strategy of~\cite{Hintz:2015}, by uniqueness of retarded solutions in $r>r_1$, $\zeta\psi$ is also equal to the restriction to $r>r_1$ of the solution $\psi_{\rm ext}$ of the extended wave equation (with complex absorption) $\mathcal{P}_{\rm ext}\psi_{\rm ext}=b'$, where $\mathcal{P}_{\rm ext}=\Box_g-i\mathcal{Q}$, on $\Omega_{\rm ext}$ with vanishing Cauchy data at $r=r_0-2\delta$ and $r=r_\sharp$.
  
  Working on Sobolev spaces, with exponential weights in $t^*$, of functions on $\Omega_{\rm ext}$ which have supported character at (i.e.\ vanish beyond) $r=r_0-2\delta$ and $r=r_\sharp$, one can then prove the Fredholm property of $\mathcal{P}_{\rm ext}\colon\{\psi_{\rm ext}\in e^{-\alpha't^*}H^s(\Omega_{\rm ext})\colon \mathcal{P}_{\rm ext}\psi_{\rm ext}\in e^{-\alpha't^*}H^{s-1}(\Omega_{\rm ext})\}\to e^{-\alpha't^*}H^{s-1}(\Omega_{\rm ext})$, where $s=s(r)$ is now a suitable variable order function subject to the bound $s<\frac12+\frac{\alpha'}{\kappa_1}$ at the Cauchy horizon $r=r_1$, by following the arguments in \cite{Hintz:2015}. The first simplification afforded by working in $r\leq r_\sharp<r_2$ is that there is no trapping in $\Omega_{\rm ext}$, which is why $\alpha'$ can indeed be taken to be arbitrary (in particular, positive) here. The second simplification is that $\mathcal{P}_{\rm ext}$ does not have any mode solutions, subject to the vanishing condition in $r>r_\sharp$, which are non-zero in $r>r_1$; this follows from domain of dependence considerations in $\{r_1<r<r_\sharp\}$ (contained in region II in Figure~\ref{fig: Penrose diagram full}).

  Having thus recovered $\zeta\psi$ as the restriction $\psi_{\rm ext}|_{r>r_1}$ of the solution of the extended equation $\mathcal{P}_{\rm ext}\psi_{\rm ext}=b'$, we can apply the radial point estimates of \cite[Proposition~2.23]{Hintz:2015} to the extended equation and deduce, upon restriction to $r>r_1$ and the region $u_\sharp<u<u_\flat$ (where the weights in $t^*\sim u\sim 1$ are irrelevant) that
  \[
    X_1\cdots X_j\psi\in H^{\frac{1}{2}+\frac{\alpha'}{\kappa_1}-\epsilon}([r_1,r_\sharp)\times(u_\sharp,u_\flat)\times\mathbb{S}^2)
  \]
  for all $j\leq m+2$ where each $X_i$ is one of the vector fields $\partial_u$, $(r-r_1)\partial_r$, $\partial_\omega$ (spherical vector fields). (We use here that $s>\frac{1}{2}+\frac{\alpha'}{\kappa_1}+m+2$.)
  
  Considering for small $\epsilon>0$ the quantity $\beta''=\min(\frac{\alpha'}{\kappa_1},1)-\epsilon\in(0,1)$, we proceed to analyze this condition, which implies
  \[
    X_1\cdots X_j\psi\in H^{\frac12+\beta''}([0,v_0)_V\times A),\qquad V:=r-r_1,\ v_0:=r_\sharp-r_1,\quad A:=(u_\sharp,u_\flat)\times\mathbb{S}^2 \subset \mathbb{R}^3,
  \]
  $j\leq m+2$, where $X_i=\partial_u,V\partial_V,\partial_\omega$. Using two derivatives along $\partial_u,\partial_\omega$, Sobolev embedding on $A$ implies $X_1\cdots X_j\psi \in \mathcal{C}^0(A;H^{\frac12+\beta''}([0,v_0)))$ for $j\leq m$, and therefore
  \[
    \psi \in \bigcap_{j=0}^m \mathcal{C}^{m-j}\bigl(A;H^{\frac12+\beta'';j}([0,v_0)\bigr),
  \]
  where we write $H^{s;j}([0,v_0))$ for the space of all $u\in H^s([0,v_0))$ so that $(V\partial_V)^i u\in H^s([0,v_0))$ for all $i\leq j$. Now, every $u\in H^{\frac12+\beta''}([0,v_0))\subset \mathcal{C}^0([0,v_0))$ has a well-defined value $u(0)$ at $V=0$. To complete the proof of the estimate~\eqref{eq:wave-CH-est}, it thus suffices to show the following 1-dimensional result (with $\beta''\in(0,1)$ and $j\in\mathbb{N}_0$):
  \[
    u \in H^{\frac12+\beta'';j}([0,v_0)) \implies |\partial_V^i(u-u(0))| \leq C_i V^{\beta''-i},\qquad 0\leq i\leq j.
  \]
  By definition of the space $H^{\frac12+\beta'';j}$, it suffices to prove this in the case $j=0$. But this follows directly from Sobolev embedding, which states that $H^{\frac12+\beta''}([0,v_0))\subset \mathcal{C}^{0,\beta''}([0,v_0))$ (H\"older space).
  
  \textbf{Second proof.} For technical simplicity, we will not operate at a sharp level of Sobolev regularity. Multiply $\rho^2\Box$ in~\eqref{eq:Boxpm} (with the `$-$' sign) by $\Delta_r$, and notice that in terms of $V:=r-r_1$ we have $\Delta_r(r)=-|\Delta_r'(r_1)|V+\mathcal{O}(V)^2$ near $V=0$, so $\Delta_r\partial_r=(-|\Delta'(r_1)|+\mathcal{O}(V))V\partial_V$. Therefore, setting $\Omega_1:=\frac{a}{r_1^2+a^2}$,
  \begin{align*}
    L := \frac{\Delta_r}{|\Delta_r'(r_1)|^2}\rho^2\Box_g &\equiv (V\partial_V)^2 + \frac{2\chi(r_1^2+a^2)}{|\Delta_r'(r_1)|}(\partial_{t_-}+\Omega_1\partial_{\varphi_-})V\partial_V \\
      &\qquad - \frac{\chi^2 V}{|\Delta_r'(r_1)|\Delta_\theta\sin^2\theta}(a\sin^2\theta\,\partial_{t_-}+\partial_{\varphi_-})^2 - \frac{V}{|\Delta_r'(r_1)|\sin\theta}\partial_\theta\Delta_\theta\sin\theta\,\partial_\theta
  \end{align*}
  modulo terms whose coefficients have at least one additional factor of $V$. Another way of viewing $L$ is that it is a differential operator constructed out of the vector fields $V\partial_V$, $\partial_{t_-}+\Omega_1\partial_{\varphi_-}$, and $V^{\frac12}\partial_\omega$ (weighted spherical derivatives) with coefficients that are smooth functions of $V^{\frac12}$ and $\mathbb{S}^2$, and which are independent of $t_-,\varphi_-$.
  
  It is convenient to rewrite this further. Recognizing that $\frac{|\Delta'_r(r_1)|}{2\chi(r_1^2+a^2)}=\kappa_1$ is the surface gravity, let us introduce the new variable
  \[
    U := e^{-\kappa_1 t_-}.
  \]
  Then $L$ is constructed out of the vector fields $V\partial_V$, $U\partial_U-\frac{\Omega_1}{\kappa_1}\partial_{\varphi_-}$, $V^{\frac12}\partial_\omega$ in the above sense, and to leading order (at $V=0$) equal to
  \begin{equation}
  \label{eq:Ledgeb}
    L \equiv \Bigl(V\partial_V - U\partial_U + \frac{\Omega_1}{\kappa_1}\partial_{\varphi_-}\Bigr) V\partial_V - \frac{\chi^2}{|\Delta_r'(r_1)|\Delta_\theta\sin^2\theta}(V^{\frac12}\partial_{\varphi_-})^2 - \frac{1}{|\Delta_r'(r_1)|\sin\theta}V^{\frac12}\partial_\theta \Delta_\theta\sin\theta\,V^{\frac12}\partial_\theta;
  \end{equation}
  and indeed its principal part is a Lorentzian signature quadratic form in these vector fields. Using an energy estimate near $V=0$, with vector field multiplier $-V^{-2\gamma}U^{-\frac{2\alpha'}{\kappa_1}}X$ (where $X$ is a future timelike linear combination of $V\partial_V-U\partial_U+\frac{\Omega_1}{\kappa_1}\partial_{\varphi_-}$ and $V\partial_V$) for a suitable weight $\gamma$ (sufficiently negative), one can thus bound the $L^2$-norms of $V\partial_V\psi$, $(U\partial_U-\frac{\Omega_1}{\kappa_1}\partial_{\varphi_-})\psi$, $V^{\frac12}\partial_\omega\psi$, and by $V$-integration also $\psi$ itself, on the full domain $\Omega$ in the space $U^{\frac{\alpha'}{\kappa_1}}V^\gamma L^2(\Omega)=e^{-\alpha' t_-}(r-r_1)^{-|\gamma|}L^2(\Omega)$ by $\|\psi\|_{e^{-\alpha' u}H^1(\Omega')}$.
  
  Since $\rho^2\Box_g\psi=0$, we also have $\rho^2\Box_g(A\psi)=0$ for all operators $A$ which commute with $\rho^2\Box_g$; such $A$ are (finite products of) $\partial_{t_-}$, $\partial_{\varphi_-}$, and the Carter operator $\mathcal{C}:=\frac{\chi^2}{\Delta_\theta\sin^2\theta}(a\sin^2\theta\,\partial_{t_-}+\partial_{\varphi_-})^2+\frac{1}{\sin\theta}\partial_\theta\Delta_\theta\sin\theta\,\partial_\theta$. The aforementioned energy estimate thus allows us to bound $A\psi$, for all such $A$, in the same weighted $L^2$-space.\footnote{A conceptually cleaner but technically considerably more involved procedure which avoids the usage of the subtle Carter operator is as follows: first, one proves higher regularity with respect to the above vector fields, and then commutes the equation $L\psi=0$ with $V\partial_V$, $U\partial_U-\frac{\Omega_1}{\kappa_1}\partial_{\varphi_-}$, and $\partial_\omega$. See \cite{Hintz:2023} for such a strategy where the role of the Cauchy horizon is played by null infinity, and the commutator vector fields are called `commutator b-vector fields'.} Since by elliptic regularity this in particular controls spherical derivatives of $\psi$, the terms in~\eqref{eq:Ledgeb} involving $V^{\frac12}\partial_{\varphi_-}$ and $V^{\frac12}\partial_\theta$ can now, due to the presence of $V^{\frac12}$, be considered to be of lower order near $V=0$. Thus, we can reduce the equation satisfied by $\psi$ to
  \begin{equation}
  \label{eq:transport}
    L_0\psi := X_{\rm in}X_{\rm out}\psi = \mathcal{O}(V^{\frac12})\psi, \qquad 
      X_{\rm in} = U\partial_U - \frac{\Omega_1}{\kappa_1}\partial_{\varphi_-} - V\partial_V, \quad X_{\rm out}=-V\partial_V.
  \end{equation}
  See Figure~\ref{fig:transport}.
  
  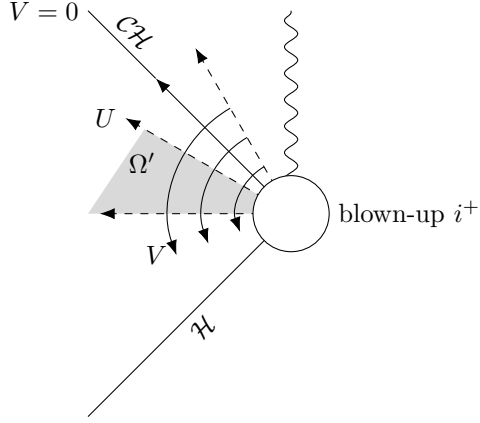
\begin{figure}
    \centering
    \begin{tikzpicture}[scale=0.9]
      \draw[snake it] (0,0) -- (0,3);

      \path[fill=gray!30](0,0)--({2.5*cos(150)},{2.5*sin(150)})--({3*cos(180)},{3*sin(180)})--(0,0);
      \draw (-2.2,0.75) node{$\Omega'$};

      \draw [-{Latex[length=2mm]}] (0,0)--(-2,2);
      \draw [dashed,-{Latex[length=2mm]}](0,0)--({sqrt(8)*cos(120)},{sqrt(8)*sin(120)});
      \draw [dashed,-{Latex[length=2mm]}](0,0)--({sqrt(8)*cos(150)},{sqrt(8)*sin(150)}) node[left]{$U$};
      \draw [dashed,-{Latex[length=2mm]}](0,0)--({sqrt(8)*cos(180)},{sqrt(8)*sin(180)});

      \draw [-{Latex[length=2mm]}] ({0.8*cos(120)},{0.8*sin(120)}) arc (120:200:0.8);
      \draw [-{Latex[length=2mm]}] ({1.3*cos(120)},{1.3*sin(120)}) arc (120:200:1.3);
      \draw [-{Latex[length=2mm]}] ({1.8*cos(120)},{1.8*sin(120)}) arc (120:200:1.8) node[left]{$V$};

      \draw (-2,2)--(-3,3) node[midway,above,sloped]{$\mathcal{CH}$} node[left]{$V=0$};

      \draw (0,0)--(-3,-3) node[midway,below,sloped]{$\mathcal{H}$};

      \draw[fill=white] (0,0) circle (16pt);
      \draw (16pt,0) node[right]{blown-up $i^+$};
    \end{tikzpicture}
    \caption{Illustration of the coordinates $U,V$, and the domain $\Omega'$. Upon blowing up $i^+$, the radial coordinate $V=r-r_1$ becomes a smooth local coordinate function down to blown-up $i^+$ (where $U=0$ in the region $r<r_2$ of validity of the coordinates $U,V$).}
    \label{fig:transport}
  \end{figure}
  
  This is a concatenation of two transport equations, which we can integrate up from a level set of $V$ in the region $\Omega'$; the following are the two main observations.
  \begin{enumerate}
  \item Integration of $X_{\rm in}$ transports decay/asymptotics of $f:=X_{\rm out}\psi$ on a hypersurface $V=V_0>0$ as $U=e^{-\kappa_1 t_-}\searrow 0$ to decay/asymptotics at the Cauchy horizon, i.e.\ as $r-r_1\searrow 0$. For example, the solution of $X_{\rm in}f(U,V,\theta,\varphi_-)=0$ with initial condition $f(U,V_0,\theta,\varphi_-)=f_0(U,\theta,\varphi_-)$ is given by
    \[
      f(U,V,\theta,\varphi_-) = f_0\Bigl(U V/V_0,\theta,\varphi_- - \frac{\Omega_1}{\kappa_1}\log(V/V_0)\Bigr).
    \]
    Note that $f_0\sim U^{\frac{\alpha'}{\kappa_1}}=e^{-\alpha' t_-}$ implies that for fixed $U=U_0$ one has $f\sim V^{\frac{\alpha'}{\kappa_1}}=(r-r_1)^{\frac{\alpha'}{\kappa_1}}$.\footnote{Likewise, for $f_0\sim(\log\frac{1}{U})^{-m}\sim t_-^{-m}$ one has $f\sim|\log(r-r_1)|^{-\frac{\alpha'}{\kappa_1}m}$ as $r\searrow r_1$ for fixed $U=U_0$. Upon integrating this along $X_{\rm in}=-(r-r_1)\partial_r$, this is a concrete manifestation of the logarithmic regularity at the Cauchy horizon on Kerr spacetimes discussed in \cite{Hintz:2017}.}
  \item Subsequent integration of $X_{\rm in}\psi=-V\partial_V\psi=f=\mathcal{O}(V^{\frac{\alpha'}{\kappa_1}})$ from $V=V_0>0$ towards the Cauchy horizon at $V=0$ produces $\psi(U,V,\theta,\varphi_-)=\psi_0(U,\theta,\varphi_-)+\psi_1(U,V,\theta,\varphi_-)$ where $\psi_1=\mathcal{O}(V^{\frac{\alpha'}{\kappa_1}})$.
  \end{enumerate}
  For a rigorous proof of~\eqref{eq:wave-CH-est}, one uses equation~\eqref{eq:transport} to improve control on bounds or asymptotic behavior of $\psi$ near $V=0$ by half a power of $V$ by controlling $L^2$- (or pointwise) norms of the integrations of $X_{\rm in},X_{\rm out}$.
\end{proof}

For the Klein--Gordon equation on subextremal RNdS or KdS spacetimes, the same pointwise bounds hold under the assumption of mode stability. In the RNdS setting, mode stability can be checked using separation of variables and a Wronskian argument for the radial ODE (see \cite[Section 1.5]{ShlapentokhRothman:2015} for the case of massless fields in the Schwarzschild case, with a scalar field mass being easily incorporated), and a full resonance expansion can be obtained using the techniques of \cite{Bony:2008,Dyatlov:2012}; in the slowly rotating KdS setting and for small scalar field masses, this is mentioned in \cite{Dyatlov:2011} and follows from \cite[Lemma~3.5]{HintzVasy:2015}, and in the small mass KdS setting it is proved in \cite{Hintz:2021}. The proof of Proposition~\ref{prop:wave-CH} goes through with only notational changes, and thus also Corollary~\ref{cor:wave-CH} holds in these settings.

We do not discuss charged scalar fields in this paper. It is known that for some values of the black hole and scalar field parameters they do not satisfy mode stability \cite{Besset:2021}.

\section{Bounding the state-dependence}
\label{sec:bound}
In the following, we focus our attention on the stress-energy tensor of the free scalar field, which is the most relevant observable of the quantum field for semi-classical gravity. We take the spacetime $\mathcal{M}$ to be a subextremal physical RNdS or KdS spacetime as described in Section \ref{sec:setup}, and we consider a scalar quantum field theory described by the CCR-algebra $\mathcal{A}(\mathcal{M})$ on this spacetime. The corresponding stress-energy tensor of the classical scalar field $\phi(x)$ is given by
\begin{align}
    T_{\mu\nu}(x)=&(1-2\xi)\partial_\mu\phi(x)\partial_\nu\phi(x)+\xi\left(R_{\mu\nu}\phi(x)^2-2\xi \phi(x)\nabla_\mu\nabla_\nu\phi(x)\right)\\\nonumber
    &-\frac{1}{2}g_{\mu\nu}\left((1-4\xi)\partial_\sigma\phi(x) \partial^\sigma\phi(x)-4\xi \phi(x)\nabla_\sigma \nabla^\sigma\phi(x)+(m^2+\xi R)\phi(x)^2\right)\, .
\end{align}
Since this is local and quadratic in the field $\phi$, the corresponding observable of the quantum field requires renormalization to be well-defined. Let us assume that the quantum stress-energy tensor is renormalized locally and covariantly utilizing Hadamard point-split renormalization \cite{Hollands:2001}, and that the renormalization ambiguities have been fixed in some way. Let us denote the resulting observable by $T^{\text{ren}}_{\mu\nu}(x)$. It is a special case of a finite sum of Wick squares of differentiated fields, with derivatives up to second order, as discussed at the end of Section~\ref{subsec: scalar field}. Therefore, if $\omega$ is a Hadamard state on $\mathcal{A}(\mathcal{M})$, then the expectation value $\omega(T^{\text{ren}}_{\mu\nu}(x))$ will be well-defined and finite for all $x\in\mathcal{M}$, but will in general diverge at the boundaries of $\mathcal{M}$ considered as a submanifold of its maximal analytic extension. We are interested in the divergence at the ingoing Cauchy horizon $\mCH$.

To study this divergence, let us fix a reference Hadamard state $\omega_0$ on $\mathcal{A}(\mathcal{M})$. One possible choice that has been used in the literature is the Unruh state \cite{Unruh:1976, Dappiaggi:2009, Brum:2014, Hollands:2019, Gerard:2020, Klein:2022}, but one could also make a different choice.
The expectation values of the components $T_{\mu\nu}^{\text{ren}}(x)$ of the renormalized stress-energy tensor in some coordinate system that is regular across $\mCH$ in the state $\omega$ can then be written as
\begin{align}
\label{eq: state dep  nsp splitting}
    \omega(T_{\mu\nu}^{\text{ren}}(x))&=\omega_0(T_{\mu\nu}^{\text{ren}}(x))+\omega(T_{\mu\nu}^{\text{ren}}(x))-\omega_0(T_{\mu\nu}^{\text{ren}}(x))\\\nonumber
   &=\omega_0(T_{\mu\nu}^{\text{ren}}(x))+\lim\limits_{x^\prime \to x}D_{\mu\nu}(x,x^\prime)W[\omega,\omega_0](x,x^\prime)\,  
\end{align}
 for $x\in \mathcal{M}$.
In the last step, we have used \eqref{eq:point-split} to rewrite the difference of expectation values. The differential operator $D_{\mu\nu}(x,x^\prime)$ can be written as
\begin{align}
    D_{\mu\nu}(x,x^\prime)=& (1-2\xi)g_{(\nu}^{\nu^\prime}(x,x^\prime)\partial_{\mu)} \partial_{\nu^\prime}+\xi\left(R_{\mu\nu}(x)-2\nabla_\mu\nabla_\nu\right)\\\nonumber
    &-\frac{1}{2}g_{\mu\nu}(x)\left((1-4\xi)g^{\sigma\rho}(x) g_{\rho}^{\rho^\prime}(x,x^\prime)\partial_\sigma \partial_{\rho^\prime}+(m^2+\xi R)-4\xi \nabla_\sigma \nabla^\sigma \right)
    \, ,
\end{align}
where $g_{\alpha}^{\beta^\prime}(x,x^\prime)$ is the bi-tensor of parallel transport, round brackets around indices indicate a symmetrization, and (un)primed derivatives act on the (un)primed variable.

The first term in \eqref{eq: state dep  nsp splitting} is independent of  the state $\omega$, and depends only on the reference state $\omega_0$. This is the state-independent part discussed before. It can be computed numerically, and indeed its quadratic leading divergence in $(r-r_1)$ at the Cauchy horizon has been found to be non-vanishing on RNdS \cite{Hollands:2019}, as well as on Reissner--Nordstr\"om \cite{Zilberman:2019} and Kerr \cite{Zilberman:2022}. The numerical results on RNdS, Kerr and Reissner--Nordstr\"om also indicate a smooth dependence of the coefficient of the $(r-r_1)^{-2}$-divergence on the spacetime parameters and the parameters of the scalar field. Moreover, first numerical results on KdS \cite{CHKS:2023} indicate that the coefficient of the $(r-r_1)^{-2}$-divergence is indeed generically nonvanishing.

The remainder of this section will be devoted to bounding the potential divergence of the state-dependent contribution of the second term in \eqref{eq: state dep  nsp splitting},
\begin{align}
   t_{\mu\nu}^\omega(x)=\lim\limits_{x^\prime \to x}D_{\mu\nu}(x,x^\prime)W[\omega,\omega_0](x,x^\prime)\, ,
\end{align}
as $x$ approaches the Cauchy horizon. In fact, one can show the following more general result:
\begin{prop}
\label{prop:t mu nu bound}
Let $x\in\mCH$ be a point on the Cauchy horizon of RNdS or KdS. Let $\mathcal{U}$ be a small open neighbourhood of $x$ with compact closure in the analytic extension of $\mathcal{M}$ and contained in the coordinate chart $(V,y^{i})$. Here, the coordinates $(y^{i})=(t_-,\theta,\varphi_-)$ parameterize $\mCH$, and $V=r-r_1$ is chosen so that $\mathcal{U}_{\mathcal{M}}:=\mathcal{U}\cap\mathcal{M}=\mathcal{U}\cap\{V>0\}$.
 Assume that the spectral gap $\alpha$ of quasinormal modes is strictly positive. Set $\beta=\alpha/\kappa_1$, and let $\mathcal{D}_{j}$, $j\in\{1,2\}$ be (not necessarily scalar) differential operators of order $m_j$, so that $m_1+m_2\leq2$, and with coefficients that are smooth on $\overline{\mathcal{M}\cup \mathcal{U}}$. Let $\omega_j$ be Hadamard states on $\mathcal{M}$, and set
 \begin{align}
   A[\omega_1,\omega_2](x)=\lim\limits_{x^\prime\to x}\left(g(x,x^\prime) \mathcal{D}_1(x) \mathcal{D}_2(x^\prime) W[\omega_1,\omega_2](x,x^\prime)\right)\, ,
 \end{align}
 where, as in \eqref{eq:point-split}, $g(x,x^\prime)$ is the proper power of the parallel transport bi-tensor, so that $A[\omega_1,\omega_2](z)$ is a $(k,l)$-tensor at $z$ for some $k,l\in \bN_0$ and all $z\in\mathcal{M}$.
 Then the tensor components $A[\omega_1,\omega_2]^{\mu_1,\dots,\mu_k}_{\nu_1,\dots,\nu_l}(V,\cdot )$ are smooth functions of $y^i$ on $\mathcal{U}_{\mathcal{M}}$ and there is a constant $C>0$ so that
\begin{align}
   \left\vert V^{2-\beta^\prime} A[\omega_1,\omega_2]^{\mu_1,\dots,\mu_k}_{\nu_1,\dots,\nu_l}(x)\right\vert \leq C
\end{align}
uniformly in $y^{i}$ within $\mathcal{U}_{\mathcal{M}}$ for some $0<\beta' <\min(\beta,1)$.
\end{prop}
(Conjecturally, one can take $\beta'$ to be arbitrarily close to $\min(\beta,1)$.) From this, the corresponding claim for $t_{\mu\nu}(x)$ follows immediately by choosing the right combination of derivative operators.

To prove Proposition~\ref{prop:t mu nu bound}, we will first show that $W[\omega,\omega_0](x,x^\prime)$, with $x$ and $x^\prime$ in $\mathcal{U}_{\mathcal{M}}$, can be rewritten as a series of forward solutions $E^+(b_i)$ to the Klein--Gordon equation with smooth and compactly supported sources $b_i\in \mathcal{C}_0^\infty(\mathcal{M})$.

\begin{lem}
\label{lem: series expansion}
Let $\mathcal{M}$ be the physical RNdS or KdS spacetime, and let $\omega_1$, $\omega_2$ be Hadamard states on $\mathcal{A}(\mathcal{M})$. Let $x,y\in\mathcal{U}_{\mathcal{M}}$ as described in Proposition \ref{prop:t mu nu bound}. Then there exists a sequence $(b_i)_{i\in \bN}\subset \mathcal{C}_0^\infty(\mathcal{M})$ of real-valued test functions satisfying 
\begin{align}
\label{eq: convergence}
    \sum\limits_i\norm{b_i}_{\mathcal{C}^m}^2=C(m)<\infty
\end{align}
 for any $m\in\mathbb{N}$ and some constants $C(m)>0$ and, in a distributional sense on $\mathcal{U}_{\mathcal{M}}\times\mathcal{U}_{\mathcal{M}}$,
\begin{align}
    W[\omega_1,\omega_2](x,y)=\sum\limits_ic_iE^+(b_i)(x)E^+(b_i)(y)\, ,
\end{align}
with $c_i=\pm1$.
\end{lem}

\begin{proof}
Recall that since both $\omega_1$ and $\omega_2$ are Hadamard states on $\mathcal{A}(\mathcal{M})$, $W[\omega_1,\omega_2](x,y)$ is a smooth, real, and symmetric function on $\mathcal{M}\times\mathcal{M}$ which solves $\mathcal{P}(x)W[\omega_1,\omega_2](x,y)=\mathcal{P}(y)W[\omega_1,\omega_2](x,y)=0$. In the rest of the proof, we will write $W(x,y)=W[\omega_1,\omega_2](x,y)$ for brevity of notation.

Let $\Sigma_\pm$ be two Cauchy surfaces of $\mathcal{M}$ to the past of $\mathcal{U}_{\mathcal{M}}$ satisfying $\Sigma_+\subset I^+(\Sigma_-)$, and define a subordinate partition of unity $(\chi_+,\chi_-)\in \mathcal{C}^\infty(\mathcal{M})$ on $\mathcal{M}$ satisfying $\chi_\pm=1$ on $J^\pm(\Sigma_\pm)$ and $\chi_\pm=0$ on $J^\mp(\Sigma_\mp)$.
Then the linear map
\begin{align}
    \mathcal{C}_0^\infty(\mathcal{U}_{\mathcal{M}})\to \mathcal{C}^\infty(\mathcal{M}), \, f\mapsto \tilde{f}=\mathcal{P}(\chi_+E(f))
\end{align}
maps test functions supported in $\mathcal{U}_{\mathcal{M}}$
to test functions supported in the closure of $J^+(\Sigma_-)\cap J^-(\Sigma_+)\cap J(\mathcal{U}_{\mathcal{M}})$. The closure of this set is a compact subset of $\mathcal{M}$ and we will call it $G$.

Moreover, the map satisfies $E(\tilde{f})=E(f)$ for any $f\in \mathcal{C}_0^\infty(\mathcal{U}_{\mathcal{M}})$. Taking into account that the kernel of $E$ as a map acting on $\mathcal{C}_0^\infty(\mathcal{M})$ is $\mathcal{P}\mathcal{C}_0^\infty(\mathcal{M})$, there must be a function $f_0\in \mathcal{C}_0^\infty(\mathcal{M})$ so that $f=\tilde{f}+\mathcal{P}f_0$. 

Let $f$, $h \in \mathcal{C}_0^\infty(\mathcal{U}_{\mathcal{M}})$. 
Then by an application of Green's second identity (i.e. integration by parts), the preceding results entail
\begin{align} 
   \int\limits_{\mathclap{\mathcal{M}\times\mathcal{M}}} W(x,y)f(x)h(y)\td vol_g(x)\td vol_g(y) =\int\limits_{\mathclap{\mathcal{M}\times\mathcal{M}}} W(x,y)\tilde{f}(x)\tilde{h}(y)\td vol_g(x)\td vol_g(y)\, .
\end{align}

Next, let $\sigma_\pm$ be a second pair of Cauchy surfaces for $\mathcal{M}$ so that $\sigma_\pm\subset I^\pm(\Sigma_\pm)$ and $\mathcal{U}_{\mathcal{M}}\subset I^+(\sigma_+)$, and let us denote $\tilde{G}=J^-(\sigma_+)\cap J^+(\sigma_-)$.
It then follows from \cite[Lemma 3.7]{Verch:1992} that there is a $B\in \mathcal{C}_0^\infty(\mathcal{M}\times\mathcal{M})$ with support contained in $\tilde{G}\times\tilde{G}$ which satisfies
\begin{align}
    \int\limits_{\mathclap{\mathcal{M}\times\mathcal{M}}} W(x,y)\tilde{f}(x)\tilde{h}(y)\td vol_g(x)\td vol_g(y)&=\int\limits_{\mathclap{\mathcal{M}\times\mathcal{M}}} B(x,y)E(\tilde{f})(x)E(\tilde{h})(y)\td vol_g(x)\td vol_g(y)\\
    &=\int\limits_{\mathclap{\mathcal{M}\times\mathcal{M}}} B(x,y)E(f)(x)E(h)(y)\td vol_g(x)\td vol_g(y)\,. 
\end{align}
Indeed, following the proof of \cite[Lemma 3.7]{Verch:1992}, for $(x,y)\in G\times G$, $B$ is of the form
\begin{align}
    B(x,y)=\mathcal{P}(x)\mathcal{P}(y)\chi(x)\chi(y)W(x,y)\, .
\end{align}
Here, $\chi\in \mathcal{C}^\infty(\mathcal{M})$ is equal to one on $I^-(\sigma_-)$, and vanishes on $I^+(\sigma_+)$. 
An illustration of the various Cauchy surfaces and relevant subsets of $\mathcal{M}$ is shown in Figure \ref{fig: illustration}.

\begin{figure}
    \centering
    \includegraphics[scale=1]{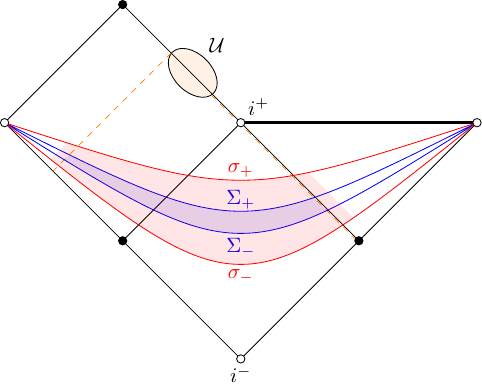}
    \caption{Illustration of the construction for the proof of Lemma \ref{lem: series expansion}. The orange ellipse represents the set $\mathcal{U}$. $J^-(\mathcal{U}_{\mathcal{M}})$ is indicated by the dashed orange lines. The blue and red hypersurfaces represent $\Sigma_\pm$ and $\sigma_\pm$, respectively. The filled blue region is the compact set $G$, the red region is $\supp(B)$.}
    \label{fig: illustration}
\end{figure}

Next, we note that $B\in \mathcal{C}_0^\infty(\tilde{G}\times\tilde{G})$ can be written in the form 
\begin{align}
    B(x,y)= \sum\limits_i c_i b_i(x) b_i(y)\,,
\end{align}
with $c_i=\pm 1$ and the $b_i\in \mathcal{C}_0^\infty(\tilde{G})$ are real-valued and satisfy \eqref{eq: convergence}.
This follows for example from \cite[App. B]{Verch:1992} and the symmetry properties of $B$.

Therefore, the functions 
\begin{align*}
    B_n(x,y)=\sum\limits_{i=0}^nc_i b_i(x) b_i(y)
\end{align*}
are integrable and bounded by an integrable function, namely the characteristic function of $\supp(B)$ multiplied by the constant $C$ from \eqref{eq: convergence}.
Thus, for $f$, $h\in \mathcal{C}_0^\infty(\mathcal{U}_{\mathcal{M}})$,
\begin{align}
    &\int\limits_{\mathclap{\mathcal{M}\times\mathcal{M}}}\sum\limits_i c_i b_i(x) b_i(y) E(f)(x)E(h)(y)\td vol_g(x)\td vol_g(y)\nonumber\\
    &= \sum\limits_i c_i \int\limits_{\mathclap{\mathcal{M}\times\mathcal{M}}} b_i(x) b_i(y) E^-(f)(x)E^-(h)(y)\td vol_g(x)\td vol_g(y)\nonumber\\
    &= \sum\limits_i c_i \int\limits_{\mathclap{\mathcal{M}\times\mathcal{M}}} \mathcal{P}(x)E^+(b_i)(x) \mathcal{P}(y)E^+(b_i)(y) E^-(f)(x)E^-(h)(y)\td vol_g(x)\td vol_g(y)\nonumber\\
    &= \sum\limits_i c_i \int\limits_{\mathclap{\mathcal{M}\times\mathcal{M}}} E^+(b_i)(x) E^+(b_i)(y) f(x)h(y)\td vol_g(x)\td vol_g(y)\nonumber\\
    &=\int\limits_{\mathclap{\mathcal{M}\times\mathcal{M}}}\sum\limits_i c_i E^+(b_i)(x) E^+(b_i)(y) f(x)h(y)\td vol_g(x)\td vol_g(y)\, ,
\end{align}
where the first equality follows from dominated convergence, the second from the properties of the fundamental solutions, the third one from Green's formula together with the compact support of $f$ and $h$ as well as the property of the fundamental solution, and the last one again from dominated convergence using the continuity of the Green operators \cite{Baer:2013}. 
This concludes the proof of the lemma.
\end{proof}

We can now prove Proposition \ref{prop:t mu nu bound}.
\begin{proof}[Proof of Proposition \ref{prop:t mu nu bound}:]
  By Lemma~\ref{lem: series expansion}, we can find a sequence of real-valued test functions $b_i\in \mathcal{C}_0^\infty(\mathcal{M})$ satisfying \eqref{eq: convergence} and with support in the red region in Figure~\ref{fig: illustration}, so that for $x\in \mathcal{U}_{\mathcal{M}}$
\begin{align}
\label{eq:decomposition}
    A[\omega_1,\omega_2]^{\mu_1,\dots,\mu_k}_{\nu_1,\dots,\nu_l}(x)(x)&=\lim\limits_{x^\prime\to x}\left(g(x,x^\prime) \mathcal{D}_1(x) \mathcal{D}_2(x^\prime)\right)^{\mu_1,\dots,\mu_k}_{\nu_1,\dots,\nu_l}W[\omega_1,\omega_2](x,x^\prime)\\\nonumber
    &=\lim\limits_{x^\prime\to x}\left(g(x,x^\prime) \mathcal{D}_1(x) \mathcal{D}_2(x^\prime)\right)^{\mu_1,\dots,\mu_k}_{\nu_1,\dots,\nu_l}(x,x^\prime)\sum\limits_ic_i\psi_i(x)\psi_i(x^\prime)\, \\\nonumber
    &=\sum\limits_i c_i \mathcal{D}_1\psi_i(x)\mathcal{D}_2\psi_i(x)\, .
\end{align}
Here we have introduced the notation $\psi_i(x)=E^+(b_i)(x)$. 
To show that the last equality holds, recall that  the continuity of the Green's operator $E^+:\mathcal{C}_0^\infty(\mathcal{M})\to \mathcal{C}^\infty(\mathcal{M})$ implies that for any compact $L\subset \mathcal{M}$ and any $m\in \bN$, there is a $m^\prime\in \bN$ and a constant $C>0$ such that 
\begin{align}
\label{eq:Cm estimate}
    \norm{\psi_i}_{\mathcal{C}^m(L)}\leq C\norm{b_i}_{\mathcal{C}^{m^\prime}}\, \quad \forall i\, .
\end{align}
Here, the $\mathcal{C}^{m^\prime}$-norm of $b_i$ is taken over a compact set containing the supports of all $b_i$. Let $K$, $K^\prime\subset \mathcal{M}$ be compact neighbourhoods of $x$ and $x^\prime$ respectively. Together with \eqref{eq: convergence}, \eqref{eq:Cm estimate} implies
\begin{align}
    \sup\limits_{K\times K^\prime} \left\vert \sum\limits_i c_i \partial_\mu \psi_i(x)\partial_\nu\psi_i(y)\right\vert \leq \sum\limits \norm{\psi_i}_{\mathcal{C}^1(K\cup K^\prime)}^2\leq C\sum\limits_i \norm{b_i}^2_{\mathcal{C}^{m^\prime(1)}}\leq \tilde{C}<\infty\, 
\end{align}
for some constant $\tilde{C}>0$.

Similar estimates can be obtained for any other number of derivatives. From these bounds one can deduce uniform convergence of the partial sums over $c_i\partial_\mu \psi_i(x)\partial_\nu \psi_i(x^\prime)$ on $K\times K^\prime$. The uniform convergence, together with the convergence of the series $\sum_i c_i\psi_i(x)\psi_i(x^\prime)$ to $W[\omega,\omega_0](x,x^\prime)$ justifies the interchange of the differentiation and the infinite sum.

Next, we consider the case where $x$ approaches the Cauchy horizon from within $\mathcal{U}_{\mathcal{M}}$. The bounds obtained in Section~\ref{sec:reg from decay}, Corollary~\ref{cor:wave-CH}, can then be used to conclude
\begin{subequations}
\label{eq:terms estimate}
\begin{align}
    \abs{\sum\limits_ic_i\partial_{y}^\gamma\psi_i(x)\partial_{y}^\delta\psi_i(x)}&\leq \sum\limits_i\abs{\partial_{y}^\gamma\psi_i(x)}\abs{\partial_y^\delta\psi_i(x)}\leq C_1 \sum\limits_i \norm{b_i}_{\mathcal{C}^{m_1}}^2\, ,\\
    \abs{V^{1-\beta'}\sum\limits_i c_i\partial_y^\gamma\psi_i(x)\partial_r \partial_y^\delta \psi_i(x)}&\leq \sum\limits_i\abs{V^{1-\beta'}\partial_r\partial_y^\delta \psi_i(x)}\abs{\partial_y^\gamma\psi_i(x)}\leq C_2 \sum\limits_i \norm{b_i}_{\mathcal{C}^{m_2}}^2\, ,\\
    \abs{V^{2-2\beta'}\sum\limits_i c_i(\partial_r\psi_i(x))^2}&\leq \sum\limits_i\abs{V^{1-\beta'}\partial_r\psi_i(x)}^2\leq C_3 \sum\limits_i \norm{b_i}_{\mathcal{C}^{m_3}}^2\, ,\\
     \abs{V^{2-\beta'}\sum\limits_i c_i\psi_i(x)\partial_r^2\psi_i(x)}&\leq \sum\limits_i\abs{V^{2-\beta'}\partial_r^2\psi_i(x)}\abs{\psi_i(x)}\leq C_4 \sum\limits_i \norm{b_i}_{\mathcal{C}^{m_4}}^2\, .
\end{align}
\end{subequations}
Here, $\gamma$, $\delta\in \bN^3\cup\{0\}$ are multi-indices, $C_j>0$, $j=1,2,3,4$, are constants, and $m_j\in \bN$. It then follows from the convergence in \eqref{eq: convergence} that the infinite sums on the right hand sides of \eqref{eq:terms estimate} are finite. Since we chose a set of coordinates in which the metric is analytically extendable across the ingoing Cauchy horizon, and since $\overline{\mathcal{U}}$ is a compact subset of the analytic extension, the smooth coefficients of the differential operators $\mathcal{D}_1$ and $\mathcal{D}_2$ are bounded on $\mathcal{U}_{\mathcal{M}}$. Combining these bounds with the ones obtained from \eqref{eq:terms estimate} then finishes the proof of the proposition. 
\end{proof}

\section{Conclusion}
\label{sec: Conclusion}
In this work we have shown that under the assumption of a positive spectral gap and mode stability, the difference of the expectation values between two Hadamard states of any quadratic observable with up to two derivatives of a real scalar quantum field on a subextremal KdS or RNdS spacetime is bounded by $(r-r_1)^{-2+\beta^\prime}$ near the Cauchy horizon. In particular, this includes the stress-energy tensor of the scalar field. As a result, a quadratic divergence of the stress-energy tensor at the Cauchy horizon in any state, such as the ones obtained numerically in \cite{Hollands:2020, Zilberman:2019, Zilberman:2022, CHKS:2023}, is, to leading order, universal in the sense of state independence. Thus, the present work lends additional importance to these numerical studies regarding the behaviour of quantum fields at the Cauchy horizon.

One can see from the method of the proof that our result can be generalized in a straightforward fashion to observables with more derivatives, and that the bound can be strengthened if the number of derivatives in the direction transversal to the Cauchy horizon (in total or on each of the fields) is lowered. Therefore, this work implies universality results of other observables like the scalar field condensate. It should also be possible to apply the arguments presented here to higher-order Wick polynomials of the scalar field, at least under the assumption that the Hadamard states involved are quasi-free.

 This work significantly extends the universality result obtained in \cite{Hollands:2019}. It suggests that perturbations caused by quantum effects can have a significant impact on the (in)stability of Cauchy horizons in asymptotically de-Sitter spacetimes. Moreover, this impact is (at leading order) independent of how the quantum field is set up. Consequently, quantum effects are likely to play an important role in fully resolving questions such as the sCC conjecture, not only in cases where the sCC can be violated classically, but in a wider variety of settings, including also the one of rotating black holes.

\section*{Acknowledgements}
C.K.~is funded by the Deutsche Forschungsgemeinschaft (DFG) under the Grant No. 406116891 within the Research Training Group RTG 2522/1. The authors would like to thank the Erwin Schr\"odinger Institut, Vienna, where part of this work has been completed, for its hospitality.

\bibliography{bibliography}
\end{document}